\documentclass[a4paper, 11pt]{article}
\usepackage[margin=3.8cm]{geometry}
\usepackage{amsmath, amsthm, amssymb}
\usepackage[round, comma]{natbib}
\usepackage{graphicx}

\newcommand{\tr}{^{\prime}}
\renewcommand{\b}{\boldsymbol}

\DeclareMathOperator{\sign}{sign}
\DeclareMathOperator{\diag}{diag}

\newtheorem{lemma}{Lemma}
\newtheorem{definition}{Definition}
\newtheorem{proposition}{Proposition}

\newtheorem{remark}{Remark}

\usepackage{color}
\definecolor{dg}{rgb}{0.1,0.5,0.1}
\title{Two algorithms for fitting constrained marginal models}

 \author{R.~J.~Evans\\
 Statistical Laboratory\\
 University of Cambridge, UK\\
 \and
 A.~Forcina\\
 Dipartimento di Economia, Finanza e Statistica\\
 University of Perugia, Italy}%
\begin{document}

\maketitle

\begin{abstract}
We study in detail the two main algorithms which have been considered for fitting constrained marginal models to discrete data, one based on Lagrange multipliers and the other on a regression model.  We show that the updates produced by the two methods are identical, but that the Lagrangian method is more efficient in the case of identically distributed observations.
We provide a generalization of the regression algorithm for modelling the effect of exogenous individual-level covariates,
a context in which the use of the Lagrangian algorithm would be infeasible for even moderate sample sizes.
An extension of the method to likelihood-based estimation under $L_1$-penalties is also considered.
\end{abstract}

%\begin{keyword}
%categorical data \sep $L_1$-penalty \sep marginal log-linear model \sep
% maximum likelihood \sep non-linear constraint.
%\end{keyword}

\begin{center}
\textbf{Keywords:}
categorical data, $L_1$-penalty, marginal log-linear model,
 maximum likelihood, non-linear constraint.
\end{center}

%\maketitle
%
\section{Introduction}
The application of marginal constraints to multi-way contingency
tables has been much investigated in the last 20 years; see, for
example, \citet{mccullagh:89, liang:92, langagresti, glonek:95, agresti:02,
BeCrHa:09}.  \citet{BerRud:02} introduced
marginal log-linear parameters (MLLPs), which generalize other
discrete parameterizations including ordinary log-linear parameters
and Glonek and McCullagh's multivariate logistic parameters. The
flexibility of this family of parameterizations enables their
application to many popular classes of conditional independence
models, and especially to graphical models \citep{FoLuMa:10, rudas:10, EvaRic:11}.
\citet{BerRud:02} show that, under certain
conditions, models defined by linear constraints on MLLPs are curved
exponential families. However, na\"ive algorithms for maximum
likelihood estimation with MLLPs face several challenges: in
general, there are no closed form equations for computing raw
probabilities from MLLPs, so direct evaluation of the log-likelihood
can be time consuming; in addition, MLLPs are not necessarily
variation independent and, as noted by \citet{BaCoFo:07}, ordinary
Newton-Raphson or Fisher scoring methods may become stuck by
producing updated estimates which are incompatible.

\citet{Lang:96} and \citet{Bergsma:97}, amongst others, have tried
to adapt a general algorithm introduced by \citet{AitSil:58} for
constrained maximum likelihood estimation to the
context of marginal models.  In this paper we provide an explicit
formulation of Aitchison and Silvey's algorithm, and show that an
alternative method due to \citet{ColFor:01} is equivalent; we term
this second approach the \emph{regression algorithm}.
Though the regression algorithm is less efficient, we show that it can be extended to deal with individual-level
covariates, a context in which Aitchison and Silvey's approach
is infeasible, unless the sample size is very small.  A variation of these algorithms, which can be
used to fit marginal log-linear models under $L_1$-penalties, and therefore perform automatic model selection, is also given.

Section \ref{sec:notation} reviews marginal log-linear
models and their basic properties, while in Section \ref{sec:algo} we formulate
the two algorithms, show that they are equivalent and discuss their properties.
In Section \ref{sec:ind} we derive an extension of the regression algorithm which can incorporate the effect of individual-level covariates.  Finally Section \ref{sec:l1} considers similar methods for $L_1$-constrained estimation.
\section{Notations and preliminary results} \label{sec:notation}
Let $X_j$, $j=1,\ldots, d$ be categorical random variables
taking values in $\{1, \dots,c_j\}$.  The joint distribution of
$X_1,\dots, X_d$ is determined by the vector of joint probabilities
$\b\pi$ of dimension $t=\prod_1^d c_j$, whose entries correspond to
cell probabilities, and are assumed to be strictly positive; we
take the entries of $\b\pi$ to be in lexographic order.
Further, let $\b y$ denote the vector of cell frequencies with
entries arranged in the same order as $\b\pi$. We write the
multinomial log-likelihood in terms of the canonical parameters as
\begin{align*}
l(\b\theta) = \b y\tr\b G\b\theta-n \log[\b 1_t\tr\exp(\b
G\b\theta)]
\end{align*}
\citep[see,
for example,][p. 699]{BaCoFo:07};
here $n$ is the sample size, $\b 1_t$ a vector of length
$t$ whose entries are all 1, and $\b G$ a $t\times (t-1)$
full rank design matrix which determines the log-linear
parameterization. The mapping between the canonical
parameters and the joint probabilities may be expressed as
\begin{align*}
\log(\b\pi)=\b G\b\theta-\b 1_t\log[\b 1_t\tr\exp(\b G\b\theta)]
\quad \Leftrightarrow\quad \b\theta=\b L\log(\b\pi),
\end{align*}
where $\b L$ is a $(t-1)\times t$ matrix of row contrasts and $\b
L\b G$ = $\b I_{t-1}$.

The score vector, $\b s$, and the expected
information matrix, $\b F$, with respect to $\b\theta$ take the form
\begin{align*}
&\b s = \b G\tr (\b y-n\b\pi)&&\mbox{and}&&\b F = n \b G\tr\b\Omega\b G;
\end{align*}
here $\b\Omega$ = $\diag(\b\pi) -\b\pi\b\pi\tr$.

\subsection{Marginal log-linear parameters}

Marginal log-linear parameters (MLLPs) enable the simultaneous modelling of
several marginal distributions \citep[see, for example,][Chapters 2
and 4]{BeCrHa:09} and the specification of suitable conditional
independencies within marginal distributions of interest
\cite[see][]{EvaRic:11}. In the following let $\b\eta$ denote an
arbitrary vector of MLLPs; it is well known that this can be written as
\begin{align*}
\b\eta = \b C\log(\b M \b\pi),
\end{align*}
where $\b C$ is a suitable matrix of row contrasts, and $\b M$ a
matrix of 0's and 1's producing the appropriate margins
\citep[see, for example,][Section 2.3.4]{BeCrHa:09}.

\citet{BerRud:02} have shown that if a vector of MLLPs $\b\eta$ is
{\em complete} and {\em hierarchical}, two properties defined below,
models determined by linear restrictions on $\b\eta$ are curved
exponential families, and thus smooth. Like ordinary log-linear
parameters, MLLPs may be grouped into interaction terms involving a
particular subset of variables; each interaction
term must be defined within a margin of which it is a subset.
\begin{definition}
A vector of MLLPs $\b\eta$ is called {\em complete} if every possible
interaction is defined in precisely one margin.
\end{definition}
\begin{definition}
A vector of MLLPs $\b\eta$ is called {\em hierarchical} if there is
a non-decreasing ordering of the margins of interest $M_1, \dots,
M_s$ such that, for each $j=1,\dots s$, no interaction
term which is a subset of $M_j$ is defined within a later margin.
\end{definition}
\section{Two algorithms for fitting marginal log-linear models}\label{sec:algo}
Here we describe the two main algorithms used for fitting models of the kind described above.
\subsection{An adaptation of Aitchison and Silvey's algorithm}
\citet{AitSil:58} study maximum likelihood estimation under
non-linear constraints in a very general context, showing that,
under certain conditions, the maximum likelihood estimates exist and
are asymptotically normal; they also outline an algorithm for
computing those estimates.  Suppose we wish to maximize $l(\b\theta)$
subject to $\b h(\b\theta)=\b 0$, a set of $r$ non-linear
constraints, under the assumption that the second derivative of $\b h(\b\theta)$ exists and is bounded.
Aitchison and
Silvey consider stationary points of the function $l(\b\theta) + \b h(\b\theta)\tr \b\lambda$, where $\b\lambda$ is a vector of Lagrange multipliers; this leads to the system of equations
\begin{equation}
\begin{split}
\b s(\hat{\b\theta}) + \b H(\hat{\b\theta}) \hat{\b\lambda} &= \b 0 \\
\b h(\hat{\b\theta}) &= \b 0,
\label{eqn:lagrange}
\end{split}
\end{equation}
where $\hat{\b\theta}$ is the ML estimate and $\b H$ the derivative
of $\b h\tr$ with respect to $\b\theta$. Since these are non-linear
equations, they suggest an iterative algorithm which proceeds as
follows: suppose that at the current iteration we have $\b\theta_0$,
a value reasonably close to $\hat{\b\theta}$.  Replace $\b s$ and
$\b h$ with first order approximations around $\b\theta_0$;
in addition replace $\b H(\hat{\b\theta})$ with $\b H(\b\theta_0)$
and the second derivative of the log-likelihood with $-\b F$, minus
the expected information matrix.
The resulting equations, after rearrangement,
may be written in matrix form as
\begin{align*}
 \left(
\begin{array}{c} \hat{\b\theta} - \b\theta_0 \\
\hat{\b\lambda}
\end{array}\right) =
\left(
\begin{array}{cc} \b F_0 & - \b H_0 \\ - \b H_0\tr & \b 0
\end{array}\right)^{-1}
\left(
\begin{array}{c} \b s_0 \\ \b h_0
\end{array}\right),
\end{align*}
where $\b s_0,\:\b F_0,\:\b H_0$ and so on denote the corresponding quantities evaluated at $\b\theta_0$.
To compute a solution, \citet{AitSil:58} exploit the structure of
the partitioned matrix, while \citet{Bergsma:97} solves explicitly
for $\hat{\b\theta}$ by substitution; in both cases, if
we are uninterested in the Lagrange multipliers, we get the
updating equation
\begin{equation}
\hat{\b\theta} = \b\theta_0+\b F_0^{-1}\b s_0-\b F_0^{-1}\b H_0(\b
H_0\tr \b F_0^{-1} \b H_0)^{-1}(\b H\tr_0 \b F_0^{-1}\b
s_0+\b h_0). \label{ASupdate}
\end{equation}
As noted by \citet{Bergsma:97}, the algorithm does not always
converge unless some sort of step length adjustment is introduced.

Linearly constrained marginal models are defined by $\b K'\b\eta$ = $\b
0$, where $\b K$ is a matrix of full column rank $r \leq t-1$. The
multinomial likelihood is a regular exponential family, so these
models may be fitted using the smooth constraint  $\b h(\b\theta)$ =
$\b K\tr \b\eta(\b \theta) = \b 0$, which implies that
\begin{align*}
\b H\tr =\frac{\partial\b h}{\partial\b\theta\tr}=\frac{\partial\b
h}{\partial\b\eta\tr} \frac{\partial\b\eta}{\partial\b\theta\tr} =\b
K\tr \b C\diag(\b M\b\pi)^{-1}\b M\diag(\b\pi)\b G.
\end{align*}
\begin{remark} \label{rmk:smooth}
In the equation above we have replaced $\b\Omega$ with $\diag(\b\pi)$
by exploiting the fact that $\b\eta$ is a homogeneous function of
$\b\pi$ \citep[see][Section 2.3.4]{BeCrHa:09}.
If the constrained model were not smooth then at
singular points the Jacobian matrix $\b R$ would not be invertible,
implying that $\b H$ is not of full rank and thus violating a
crucial assumption in \citet{AitSil:58}. It has been shown
\citep[][Theorem 3]{BerRud:02} that completeness is a necessary
condition for smoothness.

Calculation of (\ref{ASupdate}) may be simplified by noting that $\b K\tr\b C$ does not need to be updated; in addition,  if we choose, for example, $\b G$ to be the identity matrix of size $t$ with the first column removed, an explicit inverse of $\b F$ exists:
\begin{align*}
\b F^{-1} = \left[n ( \diag(\dot{\b\pi}) - \dot{\b\pi} \dot{\b\pi}\tr ) \right]^{-1} =
n^{-1} \left[\diag(\dot{\b\pi})^{-1}+\b 1_{t-1}\b 1_{t-1}\tr/(1-\b 1_{t-1}\tr \dot{\b\pi}) \right],
\end{align*}
where $\dot{\b\pi}$ denotes the vector $\b\pi$ with the first element removed; this expression may be exploited when computing $ \b F^{-1} \b H$.
\end{remark}
\subsection{A regression algorithm}\label{sec:modas}
By noting that the Aitchison-Silvey algorithm is essentially based on a quadratic approximation of $l(\b\theta)$ with a linear approximation of the constraints, \citet{ColFor:01}
designed an algorithm which they believed to be equivalent to the
original, though no formal argument was provided; this equivalence is proven in Proposition \ref{prop:equiv} below.
Recall that, by elementary linear algebra, there exists a $(t-1) \times (t-r-1)$ design matrix $\b X$ of full column rank such that $\b K' \b X = \b 0$, from which it follows that $\b\eta$ = $\b X \b\beta$ for a
vector of $t-r-1$ unknown parameters $\b\beta$. Let
\begin{align*}
\b R = \frac{\partial \b\theta}{\partial\b\eta\tr}=[\b C\diag(\b M\b\pi)^{-1}\b M\diag(\b\pi)\b G]^{-1},
\end{align*}
and $\bar{\b s}$ = $\b R\tr\b s$, $\bar{\b F}$ = $\b R\tr \b F\b R$ respectively denote the score and information relative to $\b\eta$; then the \emph{regression algorithm}
consists of alternating the following steps:
\begin{enumerate}
\item update the estimate of $\b\beta$ by
\begin{equation}
\hat{\b\beta}-\b\beta_0 = (\b X\tr\bar{\b F}_0\b X)^{-1}\b X\tr (\bar{\b F}_0 \b\gamma_0+\bar{\b s}_0),
\label{step1}
\end{equation}
where $\b\gamma_0 = \b\eta_0 - \b X \b\beta_0$;\\
\item update $\b\theta$ by
\begin{equation}
\b{\hat\theta}-\b\theta_0 =\b R_0[\b X(\hat{\b\beta}- \b\beta_0)-\b\gamma_0].
\label{step2}
\end{equation}
\end{enumerate}
\begin{proposition} \label{prop:equiv}
The updating equation in (\ref{ASupdate}) is equivalent to the combined steps given in (\ref{step1}) and (\ref{step2}).
\end{proposition}

\begin{proof}
First, consider matrices $\b X$ and $\b K$ such that the columns of $\b X$ span the orthogonal complement of the space spanned by the columns of $\b K$.  Then we claim that for any symmetric and positive definite matrix $\b W$
\begin{equation}
\b W^{-1}-\b W^{-1}\b K(\b K\tr \b W^{-1}\b K)^{-1}\b K\tr \b W^{-1}
= \b X(\b X\tr\b W\b X)^{-1}\b X\tr.\label{LAId}
\end{equation}
To see this, let $\b U$ = $\b W^{-1/2}\b K$ and $\b V$ = $\b W^{1/2}\b X$ and note that $\b U\tr \b V$ = $\b K\tr \b X=\b 0$, then (\ref{LAId}) follows from the identity $\b U(\b U\tr \b U)^{-1}\b U\tr + \b V(\b V\tr \b V)^{-1}\b V\tr$ = $ \b I$.

Now, recall $\bar{\b s}$ = $\b R\tr\b s$ and $\bar{\b F}$ = $\b
R\tr \b F\b R$, and note that
\begin{align*}
\b H\tr \b F^{-1} \b H = \b K\tr \b R^{-1} \b F^{-1}(\b R^{-1})\tr
\b K=\b K\tr \bar{\b F}^{-1}\b K;
\end{align*}
using this in the updating equation
(\ref{ASupdate}) enables us to rewrite it as
\begin{align}
\begin{split}
\b R^{-1}_0 (\b\theta-\b\theta_0) &= [\bar{\b F}_0^{-1} -\bar{\b F}_0^{-1} \b K (\b K' \bar{\b F}_0^{-1} \b K)^{-1}\b K' \bar{\b F}_0^{-1}] \bar{\b s}_0+ \\
& \qquad- \bar{\b F}_0^{-1} \b K (\b K' \bar{\b
F}_0^{-1} \b K) \b K \bar{\b F}_0^{-1}  \bar{\b F}_0 \b \eta_0.
\label{ASup}
\end{split}
\end{align}
Set $\b W$ = $\bar{\b F}_0$ and note that (\ref{LAId}) may be
substituted into the first component of (\ref{ASup}) and that its
equivalent formulation
\begin{align*}
\bar{\b F}_0^{-1}\b K(\b K\tr \bar{\b F}_0^{-1}\b K)^{-1}\b K\tr
\bar{\b F}_0^{-1} =\bar{\b F}_0^{-1}- \b X(\b X\tr\bar{\b F}_0\b
X)^{-1}\b X\tr
\end{align*}
may be substituted into the second component, giving
\begin{align*}
\b R^{-1}_0 (\b\theta-\b\theta_0) = \b X (\b X\tr \bar{\b F}_0 \b
X)^{-1}\b X\tr\bar{\b s}_0 - \b \eta_0+\b X (\b X\tr \bar{\b F}_0 \b
X)^{-1}\b X\tr \bar{\b F}_0 \b \eta_0.
\end{align*}
This is easily seen to be the same as combining equations (\ref{step1}) and (\ref{step2}).
\end{proof}

%A proof of this result is given in \ref{sec:appA}.
\begin{remark}
From the form of the updating equations (\ref{ASupdate}), (\ref{step1}) and (\ref{step2}) it is clear that Proposition \ref{prop:equiv} remains true if identical step length adjustments are applied to the $\b\theta$ updates.  This does not hold, however, if adjustments are applied to the $\b\beta$ updates of the regression algorithm.
\end{remark}
\subsubsection{Derivation of the regression algorithm}
In a neighbourhood of $\b\theta_0$, approximate $l(\b\theta)$ by a quadratic function $Q$ having the same information matrix and the same score vector as $l$ at $\b\theta_0$,
\begin{align*}
l(\b\theta) \cong Q(\b\theta) = -\frac{1}{2}(\b\theta-\b t_0)\tr \b F_0 (\b\theta-\b t_0), \quad \mbox{ where } \quad \b t_0 = \b\theta_0+\b F_0^{-1}\b s_0.
\end{align*}
Now compute a linear approximation of $\b\theta$ with respect
to $\b\beta$ in a neighbourhood of $\b\theta_0$,
\begin{equation}
\b\theta-\b\theta_0\cong \b R_0 (\b X\b\beta-\b\eta_0); \label{linap}
\end{equation}
substituting into the expression for $Q$
we obtain a quadratic function in $\b\beta$.  By adding and
subtracting $\b R_0\b X\b\beta_0$ and setting $\b\delta$ =
$\b\beta-\b\beta_0$, we have
\begin{align*}
Q(\b\beta) = -\frac{1}{2}[\b R_0 \b X\b\delta -\b R_0\b\gamma_0-\b F_0^{-1}\b s_0]\tr \b F_0 [\b R_0 \b X\b\delta -\b R_0\b\gamma_0-\b F_0^{-1}\b s_0].
\end{align*}
A weighted least square solution of this local maximization
problem gives (\ref{step1}); substitution into (\ref{linap}) gives (\ref{step2}).
\begin{remark}
The choice of $\b X$ is somewhat arbitrary because the design matrix $\b X\b A$, where $\b A$ is any non-singular matrix, implements the same set of constraints as $\b X$. In many cases an obvious choice for $\b X$ is
provided by the context; otherwise, if we are not interested in the
interpretation of $\b\beta$, any numerical complement of $\b K$ will
do.
\end{remark}
\subsection{Comparison of the two algorithms}
Since the matrices $\b C$ and $\b M$ have dimensions $(t-1) \times u$ and $u \times t$ respectively, where the value of $u \geq t$ depends upon the particular parametrization, %and the number of constraints $r$ usually grown with $t$,
the hardest step in the Aitchson-Silvey's algorithm is $(\b K\tr \b C)\diag(\b M\b\pi)^{-1}\b M$ whose computational complexity is $O(rut)$.
In contrast, the hardest step in the regression algorithm is the computation of $\b R$, which has computational complexity $O(ut^2+t^3)$, making this procedure clearly less efficient. 
However, the regression algorithm can be extended to models with individual covariates, a context in which it is usually much faster than a straightforward extension of the ordinary algorithm; see Section \ref{sec:ind}.

Note that because step adjustments, if used, are not made on the same scale, each algorithm may take a slightly different number of steps to converge.

%In any case, because the differences in computing time are not, in our experience, too large, choosing between the two is not crucial;  however we observe that the regression algorithm tends to perform better when the number of constraints $r$ is large compared to $t$, but less well when $r$ is small.

%Here we present a small numerical simulation to illustrate this phenomenon, by generating two datasets:
%a $5^4$ table with 1,500 observations and 64 zero counts and a $4^5$ table with 2,000 observations and 178 empirical zeros.
%The plots in Figure \ref{fig:comp} show the average computational time required for each algorithm in five replications, for varying numbers of constraints.
%
\subsection{Properties of the algorithms} \label{sec:prop}
Detailed conditions for the asymptotic existence of the
maximum likelihood estimates of constrained models are given by
\citet{AitSil:58}; see also \citet{BerRud:02}, Theorem 8. Much less
is known about existence for finite sample sizes where estimates
might fail to exist because of observed zeros.  In this case, some
elements of $\hat{\b\pi}$ may converge to $\b 0$, leading the
Jacobian matrix $\b R$ to become ill-conditioned and making the
algorithm unstable.

Concerning the convergence properties of their algorithm,
\citet[p.~827]{AitSil:58} noted only that it could be seen as a
modified Newton algorithm and that similar modifications had been used
successfully elsewhere. However, it is clear from the form of the
updating equations that, if the algorithms converge to some $\b\theta^*$, then the constraints $\b h(\b\theta^*) = \b 0$ are satisfied, and $\b\theta^*$ is a stationary point of the constrained likelihood.  In addition, as a consequence of the Karush-Kuhn-Tucker conditions, if a local maximum of the constrained objective function exists, then it will be a saddle point of the Lagrangian \citep[see, for example,][]{bertsekas:99}.

To ensure that the stationary point reached by the algorithm is indeed
a local maximum of the original problem, one could look at the eigenvalues of the observed information with respect to $\b\beta$: if these are all strictly positive, then we know that the algorithm has indeed converged to a local maximum.  An efficient formula for computing the observed information matrix is given in \ref{sec:appB}.
Since the log-likelihood of constrained marginal models is not,
in general, concave, it might  be advisable to apply the
algorithm to a range of starting values, in order to
check that the achieved maximum is the global one.
\subsection{Extension to more general constraints}
Occasionally, one may wish to fit general constraints on marginal probabilities without the need to define a marginal log-linear parameterization; an interesting example is provided by the \emph{relational models} of \citet{klimova:11}.  They consider constrained models of the form $\b h(\b\theta) = \b A \log (\b M \b \pi) = \b 0$, where $\b A$ is an arbitrary matrix of full row rank.  Redefine
\begin{align*}
\b K\tr =\frac{\partial \b h}{\partial \b\theta\tr} = \b A \diag(\b
M\b\pi)^{-1} \b M \b\Omega \b G
\end{align*}
and note that, because $\b A$ is not a matrix of row contrasts, $\b
h$ is not homogeneous in $\b\pi$, and thus the simplification of $\b\Omega$ mentioned in Remark \ref{rmk:smooth} does not apply. If the resulting
model is smooth, implying that $\b K$ is a matrix of full column
rank $r$ everywhere in the parameter space, it can be fitted
with the ordinary Aitchison-Silvey algorithm. We now show how
the same model can also be fitted by a slight extension of the
regression algorithm.

Let $\b\theta_0$ be a starting value and $\bar{\b K}_0$ be a right inverse of $\b K\tr$ at $\b\theta_0$; consider a first order
expansion of the constraints
\begin{align*}
\b h = \b h_0+\b K_0\tr(\b\theta-\b\theta_0) = \b K_0\tr(\bar{\b
K}_0 \b h_0+\b\theta-\b\theta_0) = \b 0
\end{align*}
and let $\b X_0$ be a matrix that spans the orthogonal complement of $\b K_0$. Then, with the same order of approximation,
\begin{align*}
\bar{\b K_0} \b h_0 + \b\theta-\b\theta_0=\b X_0\b\beta;
\end{align*}
by solving the above equation for $\b\theta-\b\theta_0$ and substituting into the quadratic approximation of the log-likelihood, we obtain an updating equation similar to (\ref{step1}):
\begin{align*}
\hat{\b\beta}-\b\beta_0=(\b X_0\tr\b F_0\b X_0)^{-1}\b X_0\tr [\b
s_0+\b F_0(\bar{\b K_0}\b h_0-\b X_0\b\beta_0)].
\end{align*}
\section{Modelling the effect of individual-level covariates} \label{sec:ind}
When exogenous individual-level covariates are available, it may be of interest to allow the marginal log-linear parameters $\b\eta$ to depend upon them as in a linear model: $\b\eta_i = \b C\log(\b M\b\pi_i)$ = $\b X_i \b\beta$; here the matrix $\b X_i$ specifies how the non-zero marginal log-linear parameters depend on individual specific information, in addition to structural restrictions such as conditional independencies.
Let $\b y_i$, $i=1,\ldots,n$, be a vector of length $t$ with a 1 in the entry corresponding to the response pattern of the $i$th individual, and all other values 0; define $\b y$ to be the vector obtained by stacking the vectors $\b y_i$, one below the other.
Alternatively, if the sample size is large and the covariates can take only a limited number of distinct values, $\b y_i$ may contain the frequency table of the response variables within the sub-sample of subjects with the $i$th configuration of the covariates; in this case $n$ denotes the number of strata. This arrangement avoids the need to
construct a joint contingency table of responses and covariates; in
addition the covariate configurations with no observations are
simply ignored.

In either case, to implement the Aitchison-Silvey approach, stack the $\b X_i$ matrices one below the other into the matrix $\b X$, and let $\b K$ span the orthogonal complement of $\b X$; as before, we have to fit the set of constraints $\b K\tr \b\eta$ = $\b 0$. However, whilst $q$, the size of $\b\beta$, does not depend on the number of subjects,  $\b H$ is now of size $[n(t-1)-q]\times n(t-1)$, and its computation has complexity $O(n^3 t^2 u)$, where $u \geq t$ as before; in addition, the inversion of the $[n(t-1)-q] \times [n(t-1)-q]$-matrix $\b H' \b F^{-1} \b H$ has complexity $O(n^3 t^3)$.   With $n$ moderately large, this approach becomes almost infeasible.

%In the regression formulation the complexity of the problem is at worst linear in $n$, as we now show.
For the regression algorithm, let $\b\theta_i$ denote the vector of canonical parameters for the
$i$th individual and
$l(\b\theta_i)$ = $\b y_i\tr\b G\b\theta_i-\log[\b 1_t\tr\exp(\b G\b\theta_i)]$
be the contribution to the log-likelihood.  Note that $\b X_i$ need not be of full column rank, a property which must instead hold for the matrix $\b X$; for this reason our assumptions are much weaker than those used by \citet{Lang:96}, and allow for more flexible models. Both the quadratic and the linear approximations must be applied at the individual level; thus we set
$\b\theta_i-\b\theta_{i0} = \b R_{i0} (\b X_i\b\beta-\b\eta_{i0})$,
 and the log-likelihood becomes
\begin{align*}
\sum_{i=1}^n l(\b\theta_i)
&\cong -\frac{1}{2}\sum_{i=1}^n [\b R_{i0} (\b X_i \b\delta - \b\gamma_{i0}) -\b F_{i0}^{-1}\b s_{i0}]\tr \b F_{i0} [\b R_{i0} (\b X_i \b\delta - \b\gamma_{i0})-\b
F_{i0}^{-1}\b s_{i0}],
\end{align*}
where $\b\gamma_{i0}$ = $\b\eta_{i0}-\b X_i\b\beta_0$,
$\b s_i$ = $\b G\tr (\b y_i-\b\pi_i)$ and $\b F_i$ = $\b
G\tr\b\Omega_i\b G$.

Direct calculations lead to the updating expression
\begin{align*}
\hat{\b\beta}-\b\beta_0 = \left(\sum_i \b X_i\tr \b
W_i\b X_i\right)^{-1}\left[\sum \b X_i\tr (\b W_i
\b\gamma_{i0}+\b R_{i0}\tr\b s_{i0})\right],
\end{align*}
where $\b W_i$ = $\b R_{i0}\tr\b G\tr \b\Omega_{i0}\b G\b R_{i0}$.
Thus, the procedure depends upon $n$ only in that we have to sum across subjects, and so the complexity is $O(n(t^2 u + t^3))$.
%{\blue This difference starkly illustrates the advantage of the regression formulation over Aitchison and Silvey's constraint based perspective.}

As an example of the utility of the method described above, consider the application to social mobility tables in \citet{DaFiFo:10}. Social mobility tables are cross classifications of subjects according to their social class (columns) and that of their fathers (rows).
The hypothesis of equality of opportunity would imply that the social class of sons is independent of that of their fathers.
Mediating covariates may induce positive dependence between the
social classes of fathers and sons, leading to the appearance of
limited social mobility; to assess this, \citet{DaFiFo:10} fitted a model in which the vector of marginal parameters for each father-son pair was allowed to depend on individual
covariates, including the father's age, the results of cognitive and
non-cognitive test scores taken by the son at school, and his
academic qualifications. The analysis, based on the UK's National Child Development Survey, included 1,942 father-son pairs classified in a $3\times 3$ table.
All marginal log-linear parameters for the father were allowed to depend on father's age, the only available covariate for fathers; the parameters for the son and the interactions were allowed to depend on all 11 available covariates.
The fitted model used 76 parameters.

\section{$L_1$-penalized parameters} \label{sec:l1}
\citet{evans:thesis} shows that, in the context of marginal
log-linear parameters, consistent model selection can be performed
using the so-called adaptive lasso.  Since the adaptive lasso uses
$L_1$-penalties, we might
therefore be interested in relaxing the equality constraints discussed
above to a penalization framework, in which we maximize the penalized
log-likelihood
\begin{align*}
\phi(\b\theta) \equiv l(\b\theta) - \sum_{j=1}^{t-1} \nu_j |\eta_j(\b\theta)|,
\end{align*}
for some vector of penalties $\b\nu = (\nu_j) \geq \b 0$.

The advantage of penalties of this form is that one can obtain
parameter estimates which are exactly zero \citep{tibshirani:96}.
Setting parameters of the form $\b\eta$ to zero corresponds to
many interesting submodels, such as those defined by conditional independences,
\citep{FoLuMa:10, rudas:10}, we can
therefore perform model selection without the need to fit many models
separately.
For now, assume that no equality constraints hold for $\b\eta$, so we
can take $\b X$ to be the identity, and $\b\beta = \b\eta$. This gives
the quadratic form
\begin{align*}
Q(\b\eta) = -\frac{1}{2}[\b R_0 (\b \eta - \b\eta_0) - \b F_0^{-1}\b
s_0]\tr \b F_0 [\b R_0 (\b \eta - \b\eta_0) - \b F_0^{-1}\b s_0]
\end{align*}
approximating $l(\b\theta)$ as before.  Then $\phi$ is approximated by
\begin{align*}
\tilde\phi(\b\eta) \equiv -\frac{1}{2}[\b R_0 (\b \eta - \b\eta_0) -
\b F_0^{-1}\b s_0]\tr \b F_0 [\b R_0 (\b \eta - \b\eta_0) - \b
F_0^{-1}\b s_0] - \sum_{j} \nu_j |\eta_j|,
\end{align*}
and we can attempt to maximize $\phi$ by repeatedly solving the
sub-problem of maximizing $\tilde\phi$.  Now, because the quadratic
form $Q(\b\eta)$ is concave and differentiable, and the absolute
value function $|\cdot|$ is concave, coordinate-wise ascent is
guaranteed to find a local maximum of $\tilde\phi$ \citep{tseng:01}.
Coordinate-wise ascent cycles through $j = 1, 2, \ldots, t-1$,
at each step minimizing
\begin{align*}
-\frac{1}{2}[\b R_0 (\b \eta - \b\eta_0) - \b F_0^{-1}\b s_0]\tr \b
F_0 [\b R_0 (\b \eta - \b\eta_0) - \b F_0^{-1}\b s_0] - \nu_j
|\eta_j|
\end{align*}
with respect to $\eta_j$, with $\eta_1, \ldots, \eta_{j-1}, \eta_{j+1}, \ldots, \eta_{t-1}$
held fixed.  This is solved just by taking
\begin{align*}
\eta_j = \sign(\check\eta)(|\check\eta| - \nu_j)_+,
\end{align*}
where $a_+ = \max\{a, 0\}$, and $\check\eta_j$ minimizes $Q$ with
respect to $\eta_j$ \citep{friedman:10}. This approach to the
sub-problem may require a large number of iterations, but it is
extremely fast in practice because each step is so simple. If the
overall algorithm converges, then by a similar argument to that of
Section \ref{sec:prop}, together with the fact that
$\tilde\phi$ has the same supergradient as $\phi$ at $\b \eta = \b
\eta_0$, we see that we must have reached a local maximum of $\phi$.

Since penalty selection for the lasso and adaptive lasso is typically
performed using computationally
intensive procedures such as cross validation, its implementation
makes fast algorithms such as the one outlined above essential.

\subsection*{Acknowledgements}

We thank two anonymous referees, the associate editor, and editor for their suggestions, corrections, and patience. 

\appendix
\section{Computation of the observed information matrix} \label{sec:appB}
\begin{lemma}
Suppose that $\b A$ is a $p\times q$ matrix, and that $\b y$, $\b b$, $\b x$ and $\b u$ are column vectors with respective lengths $q$, $p$, $k$ and $r$.
%, $\b y$ ia $q\times 1$,
%$\b b$ is $p\times 1$,  $\b x$ is $k\times 1$ and $\b u$ is $r\times 1$, x
Then if $\b A$ and $\b b$ are constant,
\begin{equation}
\frac{\partial }{\partial \b x\tr} \, \diag(\b A\b y)\b b = \diag(\b
b)\b A\frac{\partial \b y}{\partial \b u\tr} \frac{\partial \b
u}{\partial \b x\tr}.\label{Lobin}
\end{equation}
\end{lemma}

\begin{proof}
\begin{align*}
\frac{\partial }{\partial \b x\tr}\diag(\b A\b y)\b b
&=\frac{\partial }{\partial \b u\tr}\diag(\b A\b y)\b b\frac{\partial \b u}{\partial\b x\tr} \\
&= \left(\diag(\b A\b y_{u1})\b b,\:\cdots \:\diag(\b A\b y_{uh})\b b\right)\frac{\partial \b u}{\partial\b x\tr}\\
&=\left(\diag(\b b)\b A\b y_{u1},\:\cdots \:\diag(\b b)\b A \b y_{uh}\right)\frac{\partial \b u}{\partial\b x\tr}\\
&=\diag(\b b)\b A\frac{\partial \b y}{\partial \b u\tr} \frac{\partial \b
u}{\partial \b x\tr}.
\end{align*}
\end{proof}

%If $\b Y$ is a $p\times p$ non-singular matrix, a well known
%result for the derivative of its inverse is
%\begin{align*}
%\frac{\partial \b Y^{-1}}{\partial \b x\tr} = -\b Y^{-1} \frac{\partial \b Y}{\partial \b x\tr} \b Y^{-1}.
%\end{align*}

The observed information matrix is minus the second derivative of the
log-likelihood with respect to $\b\beta$, that is
\begin{align*}
-\frac{\partial}{\partial \b\beta\tr}\left[\frac{\partial l(\b\theta)}{\partial\b\beta}\right]
&= -\frac{\partial }{\partial \b\beta\tr}\b X\tr \b R\tr \b G\tr (\b y-n\b\pi)
= -\left[\frac{\partial}{\partial \b\theta\tr}\b X\tr \b R\tr \b G\tr (\b y-n\b\pi)\right] \b R \b X\\
&= n \b X\tr \b R\tr \b G\tr \b\Omega  \b G\b R \b X - \b X\tr \frac{\partial \b R\tr}{\partial \b\theta\tr} \b G\tr (\b y-n\b\pi) \b R \b X.
%=\frac{\partial\b s}{\partial\b\theta\tr} \frac{\partial\b\theta}{\partial\b\eta\tr}
%\frac{\partial\b\eta}{\partial\b\beta\tr}=\frac{\partial\b s}{\partial\b\theta\tr}\b R\b X.
\end{align*}
Since $\b s$ depends on $\b\theta$ through both $(\b y-n\b\pi)$ and $\b
R$, the above derivative has two main components, where the one
obtained by differentiating $(\b y-n\b\pi)$ is minus the expected
information.
Using the well known expression for the derivative of an inverse matrix, it only remains to compute
\begin{align*}
\b X\tr \frac{\partial \b R\tr}{\partial \b\theta\tr} \b G\tr (\b y-n\b\pi) \b R \b X &=
\b X\tr \b R\tr \frac{\partial {\b R\tr}^{-1}}{\partial \b\theta\tr} \b R\tr \b G\tr (\b y-n\b\pi) \b R \b X
= \b A \frac{\partial{\b R\tr}^{-1}}{\partial\b\theta'} \b b  \b R \b X\\
\intertext{where $\b A = \b X' \b R'$ and $\b b = \b R\tr \b G\tr (\b y-n\b\pi)$, giving}
&= \b A \b G' \frac{\partial [\diag(\b \pi) \b M' \diag(\b M \b \pi)^{-1}]}{\partial \b\theta\tr} \b C' \b b  \b R \b X.
\end{align*}
By two applications of (\ref{Lobin}), this is
\begin{align*}
\begin{split}
&\b A\b G\tr \big[\diag(\b M\tr \diag(\b M\b\pi)^{-1}\b C\tr \b
b)\\
&\qquad\qquad - \diag(\b\pi)\b M\tr \diag(\b C\tr\b b)\diag(\b M\b\pi)^{-2} \b
M \big]\b\Omega\b G  \b R \b X.
\end{split}
\end{align*}
%------------ Bibliography ----------------
\bibliographystyle{plainnat}
\bibliography{../bib}

\begin{thebibliography}{22}
\providecommand{\natexlab}[1]{#1}
\providecommand{\url}[1]{\texttt{#1}}
\expandafter\ifx\csname urlstyle\endcsname\relax
  \providecommand{\doi}[1]{doi: #1}\else
  \providecommand{\doi}{doi: \begingroup \urlstyle{rm}\Url}\fi

\bibitem[Agresti(2002)]{agresti:02}
A.~Agresti.
\newblock \emph{Categorical data analysis}.
\newblock John Wiley and Sons, 2002.

\bibitem[Aitchison and Silvey(1958)]{AitSil:58}
J.~Aitchison and S.~D. Silvey.
\newblock Maximum-likelihood estimation of parameters subject to restraints.
\newblock \emph{Ann.~Math.~Stat.}, 29\penalty0 (3):\penalty0 813--828, 1958.

\bibitem[Bartolucci et~al.(2007)Bartolucci, Colombi, and Forcina]{BaCoFo:07}
F.~Bartolucci, R.~Colombi, and A.~Forcina.
\newblock An extended class of marginal link functions for modelling
  contingency tables by equality and inequality constraints.
\newblock \emph{Statist.~Sinica}, 17\penalty0 (2):\penalty0 691, 2007.

\bibitem[Bergsma et~al.(2009)Bergsma, Croon, and Hagenaars]{BeCrHa:09}
W.~Bergsma, M.~Croon, and J.~A. Hagenaars.
\newblock \emph{Marginal Models: For Dependent, Clustered, and Longitudinal
  Categorial Data}.
\newblock Springer Verlag, 2009.

\bibitem[Bergsma(1997)]{Bergsma:97}
W.~P. Bergsma.
\newblock \emph{Marginal models for categorical data}.
\newblock Tilburg University Press, Tilburg, 1997.

\bibitem[Bergsma and Rudas(2002)]{BerRud:02}
W.~P. Bergsma and T.~Rudas.
\newblock Marginal models for categorical data.
\newblock \emph{Ann.~Statist.}, 30\penalty0 (1):\penalty0 140--159, 2002.

\bibitem[Bertsekas(1999)]{bertsekas:99}
D.P. Bertsekas.
\newblock \emph{Nonlinear programming}.
\newblock Athena Scientific, second edition, 1999.

\bibitem[Colombi and Forcina(2001)]{ColFor:01}
R.~Colombi and A.~Forcina.
\newblock Marginal regression models for the analysis of positive association
  of ordinal response variables.
\newblock \emph{Biometrika}, 88\penalty0 (4):\penalty0 1001--1019, 2001.

\bibitem[Dardanoni et~al.(2012)Dardanoni, Fiorini, and Forcina]{DaFiFo:10}
V.~Dardanoni, M.~Fiorini, and A.~Forcina.
\newblock Stochastic monotonicity in intergenerational mobility tables.
\newblock \emph{J.~Appl.~Economtrics}, 27:\penalty0 85--107, 2012.

\bibitem[Evans(2011)]{evans:thesis}
R.~J. Evans.
\newblock \emph{Parametrizations of discrete graphical models}.
\newblock PhD thesis, University of Washington, 2011.

\bibitem[Evans and Richardson(2011)]{EvaRic:11}
R.~J. Evans and T.~S. Richardson.
\newblock Marginal log-linear parameters for graphical markov models.
\newblock arXiv:1105.6075, 2011.

\bibitem[Forcina et~al.(2010)Forcina, Lupparelli, and Marchetti]{FoLuMa:10}
A.~Forcina, M.~Lupparelli, and G.~M. Marchetti.
\newblock Marginal parameterizations of discrete models defined by a set of
  conditional independencies.
\newblock \emph{Journ.~Mult.~Analysis}, 101\penalty0 (10):\penalty0 2519--2527,
  2010.

\bibitem[Friedman et~al.(2010)Friedman, Hastie, and Tibshirani]{friedman:10}
J.~Friedman, T.~Hastie, and R.~Tibshirani.
\newblock Regularization paths for generalized linear models via coordinate
  descent.
\newblock \emph{J.~Stat.~Soft.}, 33\penalty0 (1), 2010.

\bibitem[Glonek and McCullagh(1995)]{glonek:95}
G.~F.~V. Glonek and P.~McCullagh.
\newblock Multivariate logistic models.
\newblock \emph{J.~R.~Statist.~Soc.~B}, 57\penalty0 (3):\penalty0 533--546,
  1995.

\bibitem[Klimova et~al.(2011)Klimova, Rudas, and Dobra]{klimova:11}
A.~Klimova, T.~Rudas, and A.~Dobra.
\newblock Relational models for contingency tables.
\newblock arXiv:1102.5390, 2011.

\bibitem[Lang(1996)]{Lang:96}
J.B. Lang.
\newblock Maximum likelihood methods for a generalized class of log-linear
  models.
\newblock \emph{Ann.~Statist.}, 24\penalty0 (2):\penalty0 726--752, 1996.

\bibitem[Lang and Agresti(1994)]{langagresti}
J.B. Lang and A.~Agresti.
\newblock Simultaneously modeling joint and marginal distributions of
  multivariate categorical responses.
\newblock \emph{J.\ Amer.\ Statist.\ Assoc.}, 89\penalty0 (426):\penalty0
  625--632, 1994.

\bibitem[Liang et~al.(1992)Liang, Zeger, and Qaqish]{liang:92}
K.~Y. Liang, S.~L. Zeger, and B.~Qaqish.
\newblock Multivariate regression analyses for categorical data.
\newblock \emph{J.~R.~Statist.~Soc.~B}, 54\penalty0 (1):\penalty0 3--40, 1992.

\bibitem[McCullagh and Nelder(1989)]{mccullagh:89}
P.~McCullagh and J.~A. Nelder.
\newblock \emph{Generalized linear models}.
\newblock Chapman \& Hall/CRC, 1989.

\bibitem[Rudas et~al.(2010)Rudas, Bergsma, and N{\'e}meth]{rudas:10}
T.~Rudas, W.~P. Bergsma, and R.~N{\'e}meth.
\newblock Marginal log-linear parameterization of conditional independence
  models.
\newblock \emph{Biometrika}, 97\penalty0 (4):\penalty0 1006--1012, 2010.

\bibitem[Tibshirani(1996)]{tibshirani:96}
R.~Tibshirani.
\newblock Regression shrinkage and selection via the lasso.
\newblock \emph{J.~Royal Statist.~Soc.~B}, 58\penalty0 (1):\penalty0 267--288,
  1996.

\bibitem[Tseng(2001)]{tseng:01}
P.~Tseng.
\newblock Convergence of a block coordinate descent method for
  nondifferentiable minimization.
\newblock \emph{J.~Optim.~Theory Appl.}, 109\penalty0 (3):\penalty0 475--494,
  2001.

\end{thebibliography}

\end{document}